\documentclass[a4paper,UKenglish]{article}
 
\usepackage{microtype}
\usepackage{amsmath, amssymb, amsthm}
\usepackage{geometry}
\title{Tribes Is Hard in the Message Passing
  Model\footnote{A. Chattopadhyay is partially supported by a
    Ramanujan Fellowship of the DST and S. Mukhopadhyay is supported
    by a TCS Fellowship.}} 

\author{Arkadev Chattopadhyay \\\bigskip Sagnik Mukhopadhyay \\Tata Institute
  of Fundamental Research, Mumbai\\ \texttt{\{arkadev.c\ |\ sagnik\} @tifr.res.in}}

\newcommand{\h}{\mathcal{H}}
\newcommand{\bool}{\{0,1\}}
\newcommand{\tribes}{\mathsf{Tribes}_{m,\ell}}
\newcommand{\disj}{\mathsf{Disj}_\ell}
\newcommand{\Disj}{\mathsf{Disj}_2}
\newcommand{\AND}{\mathsf{AND}_k}
\newcommand{\I}{\mathbb{I}}
\newtheorem{theorem}{Theorem}
\newtheorem{definition}[theorem]{Definition}
\newtheorem{remark}[]{Remark}
\newtheorem{fact}[theorem]{Fact}
\newtheorem{observation}[theorem]{Observation}
\newtheorem{lemma}[theorem]{Lemma}
\newtheorem{claim}[theorem]{Claim}

\newcommand{\M}{\mathbf{M}}
\newcommand{\Z}{\mathbf{Z}}
\newcommand{\W}{\mathbf{W}}
\newcommand{\X}{\mathbf{X}}
\newcommand{\Y}{\mathbf{Y}}
\newcommand{\U}{\mathbf{U}}
\newcommand{\N}{\mathbf{N}}
\newcommand{\Ss}{\mathbf{S}}
\newcommand{\V}{\mathbf{V}}
\newcommand{\J}{\mathbf{J}}
\newcommand{\Pp}{\mathbf{P}}
\newcommand{\Qq}{\mathbf{Q}}
\newcommand{\Rr}{\mathbf{R}}
\newcommand{\Tt}{\mathbf{T}}

\begin{document}

\maketitle

\begin{abstract}
  We consider the point-to-point message passing model of
  communication in which there are $k$ processors with individual
  private inputs, each $n$-bit long. Each processor is located at the
  node of an underlying undirected graph and has access to private
  random coins. An edge of the graph is a private channel of
  communication between its endpoints. The processors have to compute
  a given function of all their inputs by communicating along these
  channels. While this model has been widely used in distributed
  computing, strong lower bounds on the amount of communication needed
  to compute simple functions have just begun to appear.

  In this work, we prove a tight lower bound of $\Omega(kn)$ on the
  communication needed for computing the Tribes function, when the
  underlying graph is a star of $k+1$ nodes that has $k$ leaves with
  inputs and a center with no input. Lower bound on this topology
  easily implies comparable bounds for others.  Our lower bounds are
  obtained by building upon the recent information theoretic
  techniques of Braverman et.al
  (\cite{DBLP:conf/focs/BravermanEOPV13}, FOCS'13) and combining it
  with the earlier work of Jayram, Kumar and Sivakumar
  (\cite{DBLP:conf/stoc/JayramKS03}, STOC'03). This approach yields
  information complexity bounds that is of independent interest.
 
\end{abstract}

\section{Introduction} \label{sec:introduction} The classical model of
2-party communication was introduced in the seminal work of
Yao\cite{DBLP:conf/stoc/Yao79}, motivated by problems of distributed
computing. This model has proved to be of fundamental importance (see
the book by Kushilevitz and Nisan \cite{DBLP:books/daglib/0011756})
and forms the core of the vibrant subject of communication complexity.
It is fair to say that the wide applicability of this model to
different areas of computer science cannot be over-emphasized.

However, a commonly encountered situation in distributed computing is
one where there are multiple processors, each holding a private input,
that are connected by an underlying communication graph. An edge of
the graph corresponds to a private channel of communication between
the endpoints. There are $k$ processors located on distinct nodes of
the graph that want to compute a function of their joint inputs. In
such a networked scenario, a very natural question is to understand
how much total communication is needed to get the function
computed. The classical 2-party model is just a special case where the
graph is an edge connecting two processors.

Among others, this model has also been called the Number-in-hand
multiparty point-to-point message passing model of
communication. Apart from distributed computing, this model is used in
secure multiparty computation. The study of the communication cost in
the model was most likely introduced by Dolev and Feder
\cite{DBLP:journals/siamcomp/DolevF92} and further worked on by Duris
and Rolim \cite{DBLP:journals/jcss/DurisR98}. These early works
focused on deterministic communication. There has been renewed
interest in the model because it arguably better captures many of
today's networks that is studied in various distributed models: models for
map-reduce \cite{DBLP:conf/soda/KarloffSV10,
  DBLP:conf/isaac/GoodrichSZ11}, massively parallel model for
computing conjunctive queries \cite{DBLP:conf/pods/BeameKS13,
  DBLP:conf/pods/KoutrisS11}, distributed models of learning
\cite{DBLP:journals/jmlr/BalcanBFM12} and in core distributed
computing \cite{DBLP:conf/podc/DruckerKO12}. However, there were no
known systematic techniques of proving lower bounds on the cost of
randomized communication protocols that exploited the
\emph{non-broadcast} nature of the \emph{private channels} of
communication in the model. Recently, there has been a flurry of work
developing new techniques for proving lower bounds on
communication. Phillips, Verbin and Zhang
\cite{DBLP:conf/soda/PhillipsVZ12} introduced the method of
symmetrization to prove strong bounds for a variety of
functions. Their technique was further developed in the works of
Woodruff and Zhang
\cite{DBLP:conf/stoc/WoodruffZ12,DBLP:conf/wdag/WoodruffZ13,DBLP:conf/soda/WoodruffZ14}.

All these works considered the co-ordinator model, a special case,
that was introduced in the early work of
\cite{DBLP:journals/siamcomp/DolevF92}. In the co-ordinator model, the
underlying graph has the star topology with $k+1$ nodes.  There are
$k$ leaves, each holding an $n$-bit input. Each of the $k$ leaf-nodes
is connected to the center of the star. The node at the center has no
input and is called the co-ordinator. The following two simple
observations about the model will be relevant for this work: every
function can be trivially computed using $O(nk)$ bits of communication
by having each of the $k$ players send their inputs to the
co-ordinator who then outputs the answer. It is also easily observed
that the co-ordinator model can simulate a communication protocol on
an arbitrary topology having $k$ nodes with at most a $\log k$ factor
blow-up in the total communication cost.

A key lesson learnt from our experience with the classical 2-party
model is that an excellent indicator of our understanding of a model
is our ability to prove lower bounds for the widely known
Set-Disjointness problem in the model. Indeed, as surveyed in
\cite{DBLP:journals/sigact/ChattopadhyayP10}, several new and
fundamental lower bound techniques have emerged from efforts to prove
lower bounds for this function. Further, the lower bound for
Set-Disjointness, is what drives many of the applications of
communication complexity to other domains. While the symmetrization
technique of Phillips et.al and its refinements by Woodruff and Zhang
proved several lower bounds, no strong lower bounds for
Set-Disjointness were known until recently in the $k$-processor
co-ordinator model. In this setting, the relevant definition of
Set-Disjointness is the natural generalization of its 2-party
definition: view the $n$-bit inputs of the $k$ processors as a $ k
\times n$ Boolean matrix where the $i$th row corresponds to the
Processor $i$'s input. The Set-Disjointness function outputs 1 iff
there exists a column of this matrix that has no zeroes.

In an important development, Braverman
et.al. \cite{DBLP:conf/focs/BravermanEOPV13} proved a tight
$\Omega(kn)$ lower bound for Set-Disjointness in the co-ordinator
model. Their approach is to build up new information complexity tools
for this model that is a significant generalization of the 2-party
technique of Bar-Yossef
et.al. \cite{DBLP:journals/jcss/Bar-YossefJKS04}. In this work, we
further develop this information complexity method for the
co-ordinator model by considering another natural and important
function, known as Tribes$_{m,\ell}$. In this function, the $n$-bit
input to each processor is grouped into $m$ blocks, each of length
$\ell$. Thus, the overall $k \times n$ input matrix splits up into $m$
sub-matrices $A_1,\ldots,A_m$, each of dimension $k \times
\ell$. Tribes outputs 1 iff the Set-Disjointness function outputs 1 on
each sub-matrix $A_i$. This obviously imparts a direct-sum flavor to
the problem of determining the complexity of the Tribes function in
the following sense: a naive protocol will solve Tribes by
simultaneously running an optimal protocol for Set-Disjointness on
each of the $m$ instances $A_1,\ldots,A_m$. Is this strategy optimal?

This question was answered in the affirmative for the 2-party model by
Jayram, Kumar and Sivakumar \cite{DBLP:conf/stoc/JayramKS03} when they
proved an $\Omega(n)$ lower bound on the randomized communication
complexity of the Tribes function. Their work delicately extended the
information theoretic tools of Bar-Yossef et.al
\cite{DBLP:journals/jcss/Bar-YossefJKS04}. Interestingly, it also
exhibited the power of the information complexity approach. There was
no other known technique to establish a tight lower bound on the
Tribes function\footnote{This is not surprising. Two other successful
  techniques, the discrepancy and the corruption method, both yield
  lower bounds on the non-deterministic complexity. On the other
  hand, Tribes and its complement, on $n$-bit inputs, both have only
  $\sqrt{n}$ non-deterministic complexity.}.

In this work, we show that the naive strategy for solving Tribes is
optimal also in the co-ordinator model:

\begin{theorem} \label{thm:tribes} In the $k$-processor co-ordinator
  model, every bounded error randomized protocol solving the
  $\tribes$ function, has communication cost
  $\Omega\big(m\ell k\big)$, for every $k \ge 2$.
\end{theorem}

We prove this by extending and simplifying the information complexity
approach of \cite{DBLP:conf/focs/BravermanEOPV13} and the earlier work
of \cite{DBLP:conf/stoc/JayramKS03}. It is worth noting that our
bounds in Theorem~\ref{thm:tribes} hold for all values of $k$. In
particular, this also yields a lower bound for Set-Disjointness for
all values of $k$. The earlier bound of Braverman et.al. only worked
if $k = \Omega(\log n)$.


\section{Overview \& comparison with previous work}
\label{sec:overview}
We first provide a quick overview of our techniques and
contributions. We follow this up with a more detailed description,
elaborating on the main steps of the argument.

\smallskip
\noindent {\bf Brief Summary:} Recall that the $\tribes$ function can
be written as an $m$-fold \textsf{AND} of $\disj$ instances. One
possible way to show that $\tribes$ is hard in message-passing model
is to show that any protocol evaluating $\tribes$ must evaluate all
the the $\disj$ instances. This suffices to argue that $\tribes$ is
$m$ times as hard as $\disj$. By now it is well known that information
complexity provides a convenient framework to realize such direct sum
arguments. In order to do so, one needs to define a distribution on
inputs that is entirely supported on the ones of the $m$
Set-Disjointness instances of Tribes. This was the general strategy of
Jayram et.al.\cite{DBLP:conf/stoc/JayramKS03} in the 2-party
context. However, the first problem one encounters is to define an
appropriate hard distribution and a right notion of information cost
such that Disjointness has high information cost of $\Omega(k\ell)$
under that distribution \emph{in the co-ordinator model}. This turns
out to be a delicate and involved step. Various natural information
costs do not work as observed by Phillips
et.al.\cite{DBLP:conf/soda/PhillipsVZ12}. Here, we are helped by the
work of Braverman et.al.\cite{DBLP:conf/focs/BravermanEOPV13}. They
come up with an appropriate distribution $\tau$ and an information cost
measure $\mathsf{IC}^0$. However, we face two problems in using
them. The first is that $\tau$ happens to be (almost) entirely
supported on the zeroes of Set-Disjointness. Taking ideas from
\cite{DBLP:conf/stoc/JayramKS03}, we modify $\tau$ to get a
distribution $\mu$ supported exclusively on the ones of
Set-Disjointness. Roughly speaking, to sample from $\mu$, we first
sample from $\tau$ and then pick a random column of the sampled input
and force it to all ones. Intuitively, the idea is that the all ones
column is being well hidden at a random spot. If $\tau$ was hard,
$\mu$ should also remain hard. The second problem is to appropriately
modify the information cost measure $\mathsf{IC}^0$ to $\mathsf{IC}$ so
that it yields high information complexity under $\mu$. Here, we use
an idea of \cite{DBLP:conf/stoc/JayramKS03}.

However, proving that $\mathsf{IC}$ is high for protocols when inputs
are sampled according to $\mu$ raises new technical challenges. The
first challenge is to prove a direct sum result on the information
complexity of protocols as measured by $\mathsf{IC}$. Implementing
this step is a novelty of this work, where we show roughly that
$\mathsf{IC}\big(\disj\big)$ is at least
$\Omega(\ell\cdot\mathsf{IC}\big(\Disj\big))$. For showing this, we
introduce a new information measure, $\mathsf{PIC}(f)$ which is
a lower bound on $\mathsf{IC}(f)$ and will be explained in relevat
section. The final challenge is to prove that
$\mathsf{IC}\big(\Disj\big)$ is $\Omega(k)$.  We again do that by
first simplifying some of the lemmas of
\cite{DBLP:conf/focs/BravermanEOPV13} and extending them using some
ideas from the work of \cite{DBLP:conf/stoc/JayramKS03}.

\smallskip
\noindent {\bf More Detailed Account:} Among the many possible ways to
define information cost of a protocol, the definition we work with
stems from the inherent structure of the communication model.  As
evident from the previous discussion, in the model of communication we
are interested in, the co-ordinator can see the whole transcript of
the protocol but cannot see the inputs. On the other hand, the
processors can only see a local view of the transcript - the message
that is passed to them and the message they send - along with their
respective inputs. From the point of view of the co-ordinator, who has
no input, the information revealed by the transcript about the input
can be expressed by $\I[X:\Pi(X)]$. This is small for the protocol
where the co-ordinator goes around probing each player on each
coordinate to see whether any player has $0$ in it and gives up once
she finds such a player.(We call it
Protocol \textsf{A}). It is not hard to see that the information cost
can only be as high as $O(n\log k)$ for protocol \textsf{A}. A
relevant information cost measure from the point of view of processor
$i$ is $\I[X^{-i}:\Pi^i(X)\ |\ X^i]$ which measures how much
information processor $i$ learns about other inputs from the
transcript. It turns out that this information cost is also very small
for the protocol where all the processors send their respective inputs
to the co-ordinator (We call this protocol as protocol
\textsf{B}). Here $\I[X^{-i}:\Pi^i(X)\ |\ X^i]$ is $0$ for all
$i$. What is worth noticing is that in both protocols, if we consider
the sum of the two information costs, i.e., $\I[X:\Pi(X)] + \sum_i
\I[X^{-i}:\Pi^i(X)\ |\ X^i]$, it is $\Omega(nk)$ which is the kind of
bound we are aiming for.

This cost trade-off was first observed in
\cite{DBLP:conf/soda/PhillipsVZ12} but they were unable to prove a
lower bound for $\disj$ in this model of communication. Braverman et
al \cite{DBLP:conf/focs/BravermanEOPV13} solved this problem by coming
up with the following notion of information complexity. Let
$(\mathbf{X, M, Z})$ be distributed jointly according to some
distribution $\tau$. The information cost of a protocol $\Pi$ with
respect to $\tau$ is defined as,

\begin{equation}
  \mathsf{IC}^0_\tau(\Pi) = \sum_{i \in
    [k]}\left[\underset{\tau}{\I}[\mathbf{X}^i:\Pi^i(\mathbf{X})\ |\
    \mathbf{M, Z}] + \underset{\tau}{\I}[\mathbf{M}:\Pi^i(\mathbf{X})\ |\
    \mathbf{X}^i,\mathbf{ Z}] \right]
\end{equation}

Conditioning on the auxiliary random variables $\M$ and $\Z$ serves
the following purpose: Even though the distribution $\tau$ is a
non-product distribution, it can be thought of as a convex combination
of product distributions, one for each specific values of $\M$ and
$\Z$. It is well-known by now that such convex combination facilitates
proving direct-sum like result. 

The desired properties of the distribution $\tau$ are as
follows. First, the distribution should have enough entropy to make it
hard for the players to encode their inputs cheaply and send it across
to the co-ordinator. Such an encoding is attempted in protocol
\textsf{B}. Second, the distribution should be supported on inputs
which have only a few $0$'s in each column of $\disj$. This makes sure
that the co-ordinator has to probe $\Omega(k)$ processors in each
column before he finds a $0$ in that column. This attempt of probing
was undertaken by the co-ordinator in protocol \textsf{A}. The first
property can be individually satisfied by setting each processor's
input to be $0$ or $1$ with equal probability in each column. The
second property can also be individually satisfied by taking a random
processor for each column and giving it a $0$ and giving $1$ to rest
of the processors as their inputs. Let $\Z_j$ denote the processor
whose bit was fixed to $0$ in column $j$. The hard distribution for
$\disj$ is a convex combination of these two distributions. The way it
is done is by setting a Bernoulli random variable $\M_j$ for each of
the column $j$ which acts as a switch, i.e., if $\M_j=0$ the input to
the column $j$ is sampled from the first distribution, otherwise it is
sampled from the second distribution. $\M_j$ takes value $0$ with
probability $2/3$. We define $\M=\langle \M_1,\dots,\M_\ell \rangle$ and
$\Z=\langle \Z_1,\dots,\Z_\ell \rangle$.

At this point it is interesting to go back to the definition of
$\mathsf{IC}^0$ and try to see the implication of each term in the
definition. For the coordinator,
$\sum_i\underset{\tau}{\I}[\mathbf{X}^i:\Pi^i(\mathbf{X})\ |\
\mathbf{M, Z}]$ represents the amount of information revealed about
the inputs of the processors by the transcript. For convenience, we
can assume that $\M$ is with co-ordinator. We can do this without loss
of generality as the co-ordinator can sample $O(\log k)$ inputs from
column $j$ and conclude the value of $\M_j$ from it, for any $j$. This
amount of communication is okay for us as we are trying to show a
lower bound of $\Omega(nk)$. However note that we cannot assume that
the processors have the knowledge of $\M$. Had that been the
situation, the processors would have employed protocol \textsf{A} or
protocol \textsf{B} in column $j$ depending on the value of the
$\M_j$. The value of $\I[\mathbf{X}^i:\Pi^i(\mathbf{X})\ |\ \mathbf{M,
  Z}]$, in this protocol, would have been small. So we need to make
sure that we charge the processors for their effort to know the value
of $\M$. This is taken care by the second term in the definition of
$\mathsf{IC}^0$ i.e.,
$\underset{\mu}{\I}[\mathbf{M}:\Pi^i(\mathbf{X})\ |\
\mathbf{X}^i,\mathbf{ Z}]$. Braverman et al
\cite{DBLP:conf/focs/BravermanEOPV13} used this notion of information
complexity to achieve the $\Omega(\ell k)$ lower bound for the
information cost of $\disj$ with respect to the hard distribution.

\bigskip

As mentioned before, we, however, need the hard distribution $\zeta$
for $\tribes$ to be entirely supported on $1$s of $\disj$. But the
distribution $\tau$ described above is supported on $0$s of
$\disj$. Here we borrow ideas from \cite{DBLP:conf/stoc/JayramKS03}
and design a distribution $\mu$ by selecting a random column for the
$\disj$ instances and planting an all $1$ input in it. We denote the
random co-ordinate by $\W$. It is easy to verify that $\mu$ is a
distribution supported in $1$s of $\disj$. We set the hard
distribution for $\tribes$ to be an $m$-fold product distribution
$\zeta = \mu^m$ denoted by the random variables $\langle \bar{\X},
\bar{\M}, \bar{\Z}, \bar{\W} \rangle$. It is to be noted that a
correct protocol should work well for all inputs, not necessarily for
the inputs coming from the distribution $\zeta$. This property will be
crucially used in later part of the proof. The modification of the
input distribution from $\tau$ to $\mu$ and subsequently to $\zeta$
calls for changing the definition of the information complexity to
suit our purpose. We define information complexity as follows which we
will use in this paper.

\begin{definition}
  Let $(\bar{\X}, \bar{\M}, \bar{\Z}, \bar{\W})$ be distributed jointly according to
  $\zeta$. The information cost of a protocol $\Pi$ with $k$ processors in
  NIH point-to-point coordinator model with respect to $\zeta$ is
  defined as,
  \begin{equation}
    \mathsf{IC}_\zeta(\Pi) = \sum_{i \in
      [k]}\left[\underset{\zeta}{\I}[\bar{\X}^i:\Pi^i(\bar{\X})\ |\
      \bar{\M}, \bar{\Z}, \bar{\W}] + \underset{\zeta}{\I}[\bar{\M}:\Pi^i(\bar{\X})\ |\
      \bar{\X}^i,\bar{\Z}, \bar{\W}] \right].
  \end{equation}
For a function $f: \mathbf{X} \rightarrow \mathcal{R}$, the
information complexity of the function is defined as,
\begin{equation}
  \mathsf{IC}_{\zeta,\delta}(f) = \inf_\Pi \mathsf{IC}_\mu(\Pi),  
\end{equation}
where the infimum is taken over all $\delta$-error protocol $\Pi$ for
$f$.
\end{definition}

By doing this, we are able to bound the information complexity of
$\tribes$ as $m$-times that of $\disj$. Although non-trivial, this
step can be accomplished by exploiting the proof techniques used in
\cite{DBLP:conf/focs/BravermanEOPV13}. The next step is to bound the
information complexity of $\disj$, which turns out to be difficult for
two reasons. First, the distribution $\mu$ is no more a $0$
distribution for $\disj$. We get around this by defining a new
information complexity measure, - which we call as partial information
complexity - to show that the partial information complexity of
$\disj$ on distribution $\mu$ is at least $(\ell-1)$-times that of
$\Disj$. This is one of the main technical contributions of our
paper. See Section \ref{sec:directsum} for details. The second hurdle
we face is bounding the information complexity of $\Disj$. Here we
combine ideas from \cite{DBLP:conf/stoc/JayramKS03,
  DBLP:conf/focs/BravermanEOPV13} to conclude that the partial
information complexity of $\Disj$ is $\Omega(k)$. This is the second
main technical contribution of this paper, which is explained in
Section \ref{sec:lower-bounding-disj}. Finally we give a simple
argument in Section \ref{sec:corr-inform-meas} to show that
$\mathsf{IC}_\zeta(\Pi)$ lower bounds the communication cost of $\Pi$
where $\Pi$ is any correct protocol for $\tribes$.

\section{Preliminaries}
\label{sec:preliminaries}



\noindent

\textbf{Communication complexity. } In this work, we are mainly
interested in multiparty communication \textit{number-in-hand}
model. In this model of computation, the input is distributed between
$k$ players $P_1,\cdots,P_k$ who jointly wish to compute a function
$f$ on the combined input by communication with each other. It can be
assumed that the players have unlimited computational power. Several
variants of this model have been studied extensively, such as, message
passing model, where each pair of players have a dedicated
communication channel and hence the players can send messages to
specific players. This is contrasted in the second variant where the
players can only broadcast their communication. The latter model in
known as shared blackboard model. In this work, we consider the
message passing model.

As mentioned before, the model which is easier to work with is the
coordinator model where, in addition to $k$ players who hold the
input, a coordinator is introduced who does not have any input but all
the communication is channelled via her, i.e., the players can only
communicate with the coordinator though the coordinator is allowed to
communication with everybody. It is easy to observe that this
coordinator model can simulate the message passing model with only a
$\log k$ overhead in the total communication cost.

We work with randomized protocol where the players have access to
private coins. (Though it might seem like that the public coin
protocol can yield better upper bound, it can be noted that all the
proofs can be modified to give the same result for public coin model.)
The standard notion of private coin randomized communication
complexity is adopted here, where we look at the worst-case
communication of the protocol when the protocol is allowed to make
only $\delta$ error (bounded away from $1/2$) on each input. Here the
probability is taken over the private coin tosses of the players. For
more details, readers are referred to \cite{DBLP:books/daglib/0011756}.

\textbf{Information theory. } We will quickly go through the
information theoretic definitions and facts we need. For a random
variable $X$ taking value in the sample space $\Omega$ according to
the distribution $p(\dot)$, the entropy of $X$, denoted as $\h(X)$, is
defined as follows.
\begin{equation}
\label{eq:44}
  \h(X) = \sum_{x \in \Omega}\Pr[X=x]\log\frac{1}{\Pr[X=x]} = \mathbf{E}_x\left[\log\frac{1}{p(x)}\right].
\end{equation}

For two random variables $X$ and $Y$, the conditional entropy of $X$
given $Y$ is defined as follows.
\begin{equation}
  \label{eq:45}
  \h(X|Y) = \mathbf{E}_{x,y}\left[\log\frac{1}{p(x|y)}\right].
\end{equation}
Informally, the entropy of a random variable measures the uncertainty
associated with it. Conditioning on another random variable, i.e.,
knowing the value that another random variable takes can only decrease
the uncertainty of the former one. This notion is captured in the
following fact that $\h(X|Y) \leq \h(X)$ where the equality is achieved
when $X$ is independent of $Y$. Given two random variables $X$ and $Y$
with joint distribution $p(x,y)$ we can talk about how much
information one random variable reveals about the other random
variable. The mutual information, as it is called, between $X$ and $Y$
is defined as follows.
\begin{equation}
  \label{eq:46}
  \I[X:Y] = \h(X) - \h(X|Y).
\end{equation}
It is to be noted that the mutual information is a symmetric quantity,
though it might not be obvious from the definition itself. From the
previous discussion, it is easy to see that the mutual information is
a non-negative quantity. As before, we can also define conditional
mutual information as below.
\begin{equation}
  \label{eq:47}
  \I[X:Y|Z] = \h(X|Z) - \h(X|Y,Z).
\end{equation}
The following chain rule of mutual information will be crucially used
in our proof.
\begin{equation}
  \label{eq:48}
  \I[X_1,\dots,X_n:Y] = \sum_{i \in [n]}\I[X_i:Y|X_{i-1},\dots,X_1].
\end{equation}
It is to be noted that the chain rule of mutual information will also
work when conditioned on random variable $Z$.
\begin{remark}
  Consider a permutation $\sigma:[n] \rightarrow [n]$. The following
  observation will be useful in our proof.
  \begin{equation}
    \label{eq:49}
    \I[X_1,\dots,X_n:Y] = \sum_{i \in [n]}\I[X_{\sigma(i)}:Y|X_{\sigma(i-1)},\dots,X_{\sigma(1)}].
  \end{equation}
\end{remark}
We will use the following lemmas
regarding mutual information.

\begin{lemma}
  \label{lemma:condn-excg}
Consider random variables $A, B, C$ and $D$. If $A$ is independent of
$B$ given $D$ then,
\begin{equation}
  \label{eq:30}
  \I[A:B,C\ |\ D] = \I[A:C\ |\ B,D],
\end{equation}
and
\begin{equation}
  \label{eq:31}
  \I[A:C\ |\ B,D] \geq \I[A:C\ |\ D].
\end{equation}
\end{lemma}

\begin{proof}
  By chain rule of mutual information, 
  \begin{align*}
    \label{eq:32}
    \I[A:B,C\ |\ D] &=  \I[A:B\ |\ D] + \I[A:C\ |\ B,D]\\
                           &=\I[A:C\ |\ B,D] \ \ \ \ \ \text{[As
                             $A\perp B|D \Rightarrow \I[A:B\ |\ D]=0$].}
  \end{align*}

For the second expression, we know $\I[A:B|D]=0$ which means $\h(A|D)
= \h(A|B,D)$. Hence,

\begin{eqnarray}
  \label{eq:33}
   \I[A:C\ |\ B,D] &=& \h(A|B,D) - \h(A|B,C,D) \cr
                          &=& \h(A|D) - \h(A|B,C,D) \cr
                          &\geq& \h(A|D) - \h(A|C,D) \cr
                          &=& \I[A:C\ |\ D]. \nonumber
\end{eqnarray}
\end{proof}

\section{Lower bound for $\tribes$ in message-passing model}
\label{sec:lower-bound-tribes-mp}

Here, in the first subsection, we will show two direct-sum results. In
the first step we bound the information complexity of $\tribes$ in
terms of that of $\disj$. It is to be noted that the proof technique
of \cite{DBLP:journals/jcss/Bar-YossefJKS04} falls short of proving
any lower bound on the information complexity measure we have defined
- mainly because of the fact the information complexity measure
consists of sum two different mutual information terms for each
processor, and it is not clear that one can come up with lower bounds
for both the terms simultaneously. This problem has already been
attended to in \cite{DBLP:conf/focs/BravermanEOPV13} and the proof we
present here resembles the proof technique used by them. For
completeness we include the proof in this paper. In the second step,
we will bound the information complexity of $\disj$ in terms of
$\Disj$. This step is more difficult and a straight-forward
application of the direct-sum argument of
\cite{DBLP:conf/focs/BravermanEOPV13} will not work. First we use
ideas from \cite{DBLP:conf/stoc/JayramKS03} to define partial
information complexity measure which is more convenient to work
with. Then we come up with a novel direct-sum argument for partial
information complexity measure.

In Section \ref{sec:lower-bounding-disj}, we show that the
information complexity of $\Disj$ is $\Omega(k)$. We manage to show
this by combining ideas from \cite{DBLP:conf/focs/BravermanEOPV13,
  DBLP:conf/stoc/JayramKS03}.

\subsection{Direct sum}
\label{sec:directsum}
In this section we prove that the information cost of computing
$\tribes$ is $m$ times the information cost of computing $\disj$. The
proof is almost the same proof as in
\cite{DBLP:conf/focs/BravermanEOPV13} where the authors have used a
direct sum theorem to show that the information cost of computing
$\disj$ is $\ell$ times the information cost of computing $k$-bit
$\AND$. Before going into details we need the following definitions
%

Consider $f:\mathcal{D}^m \rightarrow \mathcal{R}$ can be
written as $f(X)=g(h(X_1),\dots,h(X_m))$ where $X = \langle
X_1,\dots,X_m \rangle$, $X_i \in \mathcal{D}$ and $h:\mathcal{D}
\rightarrow \mathcal{R}$. In other words, $f$ is $g$-decomposable with
primitive $h$.

\begin{definition}[Collapsing distribution]
  We call $X \in \mathcal{D}^m$ be a collapsing input for $f$ if for
  any $i \in [m]$ and $y \in \mathcal{D}$, we have
  $f(X(i,y))=h(y)$. Any distribution $\zeta$ supported
  entirely on collapsing inputs on $f$ is called a collapsing
  distribution of $f$.
\end{definition}

\begin{definition}[Projection]
 Given a distribution $\nu$ specified by random variable $(D_1,\dots,D_k)$ and a subset $S$ of $[k]$, we call the projection of $\nu$ on $(D_i)_{i\in S}$, denoted as 
 $\nu\downarrow_{(D_i)_{i\in S}}$, the marginal distribution of $(D_i)_{i\in S}$ induced by $\nu$.
\end{definition}

The proof is by reduction: we will show that given a
protocol $\Pi$ for $\tribes$ and a collapsing distribution $\mu =
\zeta_l^m$ , we can construct a protocol $\Pi'$ for
$\disj$ such that it computes $\disj$ with the same error probability
as that of $\Pi$ and the information complexity of $\Pi$ is
$m$ times that of $\disj$.

\begin{theorem}
  \label{thm:direct-sum1}
Let $\mu = \zeta_\ell^m$ be a collapsing distribution for $\tribes$ partitioned by
$\mathbf{M, Z}$ and $\mathbf{W}$ as described before. Then
\begin{equation}
  \label{eq:4}
  \mathsf{IC}_\mu(\tribes) \geq m. \mathsf{IC}_{\zeta_\ell}(\disj).
\end{equation}
\end{theorem}

As mentioned before, the proof of Theorem \ref{thm:direct-sum1} works
out nicely by adapting the proof techniques of
\cite{DBLP:conf/focs/BravermanEOPV13} and is given below.
 We will describe how to construct protocol $\Pi'$ for $\disj$
given protocol $\Pi$ for $\tribes$. On an input $u$
for $\disj$, the processors and the coordinator sample a random $k \times m\ell$ matrix $\bar{\X}$  in the following way.
\begin{enumerate}
\item The coordinator samples $\mathbf{J}$ uniformly at random from $[m]$. This
  is the place where the processors embed the input $u$. 

\item The coordinator samples $\bar{\Z}_{-\mathbf{J}} \in
  [k]^{\ell(m-1)}$ and sends it to all the processors.

\item The coordinator samples $\bar{\W}_{-\mathbf{J}} \in [\ell]^{m-1}$
  and sends it to all the processors.

\item The coordinator samples $\mathbf{M}_t \in \bool^{\ell}, \mathbf{M}_t \sim \mathsf{Bin}(1/3)$ for all $t$
  less that $\mathbf{J}$ and sends it to all the processors. The processors
  use their private randomness to sample the inputs for those $\disj$
  instances in the following way: For the $i$th $\disj$ instance, the processors sample $\X$ by sampling each of the columns independently - For column $j$ the input of the $\Z_j$th processor is fixed to $0$ and the other processors get $1$ if $\M_j=1$, otherwise, they get $0$ or $1$ uniformly at random 

\item For the rest of $\disj$ instances, the coordinator herself
  samples the inputs in the same way and sends the requisite inputs to respective
  processors.
\end{enumerate}

\begin{observation}
	\label{OBS-2}
Consider the tuple $(\U,\N,\V,\Ss)$ distributed according to $\zeta_l$. If $\U$ is 
given as input to protocol $\Pi'$, then $(\bar{\X}, \bar{\M}, \bar{\Z}, \bar{\W})$ is distributed according to $\mu$.
\end{observation}

As $\mu$ is a collapsing distribution for $\tribes$, it is easy to see
that the protocol $\Pi'$ computes $\disj$. We need to show the
connection between information cost of this protocol and the
information cost of $\tribes$ which is shown next.

We will prove the following lemma.

\begin{lemma}\cite{DBLP:conf/focs/BravermanEOPV13}
  \label{lemma:directsum1}
  \begin{equation}
    \label{eq:1}
    \underset{(\mathbf{U,N,V,S}) \sim
      \zeta_\ell}{\I}[\mathbf{U}^i:\Pi'^i(\mathbf{U})\ |\ \mathbf{N,V,S}]
    \leq \frac{1}{m}\underset{(\bar{\X},\bar{\M},\bar{\W},\bar{\Z}) \sim
      \mu}{\I}[\bar{\X}^i:\Pi^i(\bar{\X})\ |\ \bar{\M},\bar{\W},\bar{\Z}],
  \end{equation}
and
  \begin{equation}
    \label{eq:2}
    \underset{(\mathbf{U,N,V,S}) \sim
      \zeta_\ell}{\I}[\mathbf{N}:\Pi'^i(\mathbf{U})\ |\
    \mathbf{U}^i,\mathbf{V,S}]
    \leq \frac{1}{m}\underset{(\bar{\X},\bar{\M},\bar{\W},\bar{\Z}) \sim
      \mu}{\I}[\bar{\M}:\Pi^i(\bar{\X})\ |\ \bar{\X}^i,\bar{\W},\bar{\Z}].
  \end{equation}
\end{lemma}
Lemma \ref{lemma:directsum1} implies Theorem \ref{thm:direct-sum1}.

\begin{proof}[Proof of Lemma \ref{lemma:directsum1}]
  The proof is exactly same as the proof of Lemma 4.1 of
  \cite{DBLP:conf/focs/BravermanEOPV13}.

  We will prove Equation \eqref{eq:2} first. Equation \eqref{eq:1} can
  be proved similarly. Let's look at the view of processor $i$ on the
  transcript $\Pi'$, which we are calling as $\Pi'^i$. We have
  \begin{equation*}
    \Pi'^i(\mathbf{U}) = \langle \mathbf{J}, \bar{\Z}_{-\J}, \bar{\M}_{[1,\J-1]},
    \bar{\X}_{[\J+1,m]}^i, \bar{\W}_{-\J}, \Pi(\bar{\X}(
    \mathbf{J, U})) \rangle.
  \end{equation*}

  As $\langle \mathbf{J}, \bar{\Z}_{-\mathbf{J}}, \bar{\M}_{[1,\mathbf{J}-1]},
  \bar{\X}_{[\mathbf{J}+1,m]}^i, \bar{\W}_{-\mathbf{J}} \rangle$ is independent of
  $\mathbf{N}$, we can write what follows.

  \begin{eqnarray}
    \label{eq:3}
    & & \I_{(\mathbf{U,N,V,S}) \sim \zeta_\ell}[\mathbf{N}:\Pi'^i(\U)\ |\
    \U^i,\V, \mathbf{S}]\nonumber \\ \nonumber
    &=& \underset{\underset{(\bar{\Z}_{-\mathbf{J}}, \bar{\M}_{-\mathbf{J}},
        \bar{\X}_{-\mathbf{J}}, \mathbf{W}_{-\mathbf{J}}) \sim
        \zeta_\ell^{m-1}}{(\U,\N,\V,\Ss) \sim \zeta_\ell}}{\I}[\N:\langle \J, \bar{\Z}_{-\J}, \bar{\M}_{[1,\J-1]},
    \bar{\X}_{[\J+1,m]}^i, \bar{\W}_{-\J},
    \Pi(\bar{\X}(
    \mathbf{J, U}))\rangle\ |\ \U^i,\V,\mathbf{S}]\\ \nonumber
    &=& \underset{\underset{(\bar{\Z}_{-\mathbf{J}}, \bar{\M}_{-\mathbf{J}},
        \bar{\X}_{-\mathbf{J}}, \mathbf{W}_{-\mathbf{J}}) \sim
        \zeta_\ell^{m-1}}{(\U,\N,\V,\Ss) \sim \zeta_\ell}}{\I}[\N:\Pi(\bar{\X}(
    \mathbf{J, U}))\ |\  \J, \bar{\Z}_{-\J}, \bar{\M}_{[1,\J-1]},
    \bar{\X}_{[\J+1,m]}^i, \bar{\W}_{-\J}, \U^i,\V,\Ss]\\ \nonumber
    &=& \underset{(\bar{\X},\bar{\M},\bar{\W},\bar{\Z}) \sim \mu}{\I} [\bar{\M}_{\J}:
    \Pi^i(\bar{\X})\ |\ \J, \bar{\M}_{[1,\J-1}, \bar{\X}_{[\J,m]}^i,
    \bar{\Z}, \bar{\W}] \\ \nonumber
    &\leq& \underset{(\bar{\X},\bar{\M},\bar{\W},\bar{\Z}) \sim \mu}{\I} [\bar{\M}_{\J}:
    \Pi^i(\bar{\X})\ |\ \J, \bar{\M}_{[1,\J-1}, \bar{\X}^i,
    \bar{\Z}, \bar{\W}]\\ \nonumber
    &=& \frac{1}{m} \sum_{j=1}^m \underset{(\bar{\X},\bar{\M},\bar{\W},\bar{\Z}) \sim \mu}{\I} [\bar{\M}_{j}:
    \Pi^i(\bar{\X})\ |\ \bar{\M}_{[1,j-1}, \bar{\X}^i,
    \bar{\Z}, \bar{\W}]\\ \nonumber
    &=& \frac{1}{m}\underset{(\bar{\X},\bar{\M},\bar{\W},\bar{\Z}) \sim \mu}{\I} [\bar{\M}:
    \Pi^i(\bar{\X})\ |\  \bar{\X}^i,
    \bar{\Z}, \bar{\W}].\\ \nonumber
  \end{eqnarray}
  Equation \eqref{eq:1} can be proved in the same way.
\end{proof}

Now our goal is  to connect the information cost of $\disj$ under
$\zeta_\ell$ to information cost of $\AND$. So a natural attempt is to
prove a theorem like Theorem \ref{thm:direct-sum1} for reduction from $\disj$
to $\AND$. Unfortunately this is not possible. Recall that $\disj(X)=
\bigvee_{i=1}^\ell \bigwedge_{j=1}^k X_i^j$. Hence for a collapsing
distribution each of the $\AND$s should evaluate to $0$, which is not
the case for the distribution $\zeta_\ell$. 

Inspired by \cite{DBLP:conf/stoc/JayramKS03}, we define the
following measure of information cost, namely, partial
information cost. Let $\Pi$ be a protocol for $\disj$. The partial
information cost of $\Pi$ is defined as,

\begin{equation}
  \label{eq:5}
  \mathsf{PIC}(\Pi)=\sum_{i=1}^k \left( \mathbb{I}\left[\M_{-\W}:\Pi^i(\X)\ |\ \X^i, \Z,
      \W\right] + \mathbb{I}\left[\X^i_{-\W}:\Pi^i(\X)\ |\ \M, \Z,
      \W\right] \right).
\end{equation}
The random variable $\M_{-\W}$ denotes $\M$ with its $\W$-th
coordinate removed. Similarly, $\X^i_{-\W}$ denotes $\X^i$ with its
$\W$-th coordinate removed. The partial information complexity of
$\disj$ is the partial information cost of the best protocol computing
$\disj$. It is easy to see that the partial information complexity of
any function $f$ lower bounds the information complexity of $f$.

We prove the following theorem.

\begin{theorem}
  \label{thm:moddirectsum}
Let $\zeta_\ell$ be the distribution over the inputs of $\disj$
partitioned by $\M,\Z,\W$ as described before. Then
\begin{equation}
  \label{eq:15}
  {\mathsf{PIC}}_{\zeta_\ell}(\disj) \geq (\ell-1). {\mathsf{PIC}}_{\zeta_2}(\Disj).
\end{equation}
\end{theorem}

Here we will show the following reduction analogous to our previous
reduction from $\tribes$ to $\disj$. Given a protocol $\Pi'$ for
$\disj$ and distribution $\zeta_\ell$ (as described in Section
\ref{sec:overview}, we will come up with a protocol $\Pi''$ for
$\Disj$ such that the partial information cost of $\Pi''$
w.r.t. $\zeta_\ell$ is $1/(\ell-1)$ times the partial information cost
of $\Pi'$ w.r.t $\zeta_2$.

Let us describe the construction of the protocol $\Pi''$. On an input
$u = \langle u_1, u_2 \rangle$ for $\Disj$, the processors and the
coordinator sample a $k \times \ell$ random matrix $\X(u)$ in the following way.

\begin{enumerate}
\item The coordinator samples $\Pp$ and $\Qq$ uniformly at random from
  $[\ell]$ such that $\Pp < \Qq$.
\item The coordinator samples $\Z_{-\{\Pp,\Qq}\}=(\Z_i)_{i \in
    [\ell]\backslash \Pp,\Qq}$, where each $\Z_i
  \underset{R}{\in}[k]$, and sends it to all the processors.
\item The coordinator samples a number $\Rr$ uniformly at random from
  $\{0,...,\ell-2\}$ and then samples a subset $\Tt \subseteq
  [\ell]\backslash\{\Pp,\Qq\}$ uniformly at random from all sets of
  size $\Rr$ that do not contain $\Pp, \Qq$. Then the coordinator
  samples $\M_t \sim \mathsf{Bin}(1/3)$ for all $t \in \Tt$ and sends
  them to all the processors. The processors use their private
  randomness to sample $\X_t$ for each column $t$ in $\Tt$ in the
  following way: The input of the $\Z_t$-th processor is fixed to $0$
  in $\X_t$ and the other processors get $1$ if $\M_t=1$, otherwise,
  if $\M_t=0$, they get $0$ or $1$ uniformly at random. We will call
  this input sampling procedure as \textsf{IpSample}.
\item For the rest of the columns, the coordinator samples the inputs
  according to \textsf{IpSample} and sends the requisite inputs to the
  respective processors.
\item The processors form the input $\X\equiv \X(u,\Pp,\Qq)$
  (i.e., $\X_{\Pp}=u_1$ and $\X_{\Qq}=u_2$) and run the protocol
  $\Pi'$ for $\disj$ with $\X$ as input.
\end{enumerate}

\begin{observation}
	\label{OBS-1}
Consider the tuple $(\U,\N,\V,\Ss)$ distributed according to $\zeta_2$. If $\U$ is 
given as input to protocol $\Pi''$, then $(\X, \M, \Z, \W)$ is distributed according to $\zeta_\ell$, where $\W$ is the unique all 1's coordinate in $\X$. Here $\W=\Pp$ if $\V=1$ and $\W=\Qq$ if $\V=2$.
\end{observation}

Next we prove the following lemma connecting the information cost of
$\Pi'$ for $\disj$ and that of $\Pi''$ for $\Disj$. This lemma implies the Theorem \ref{thm:moddirectsum}.

\begin{lemma}
  \label{lemma:mod-direct-sum}
  \begin{equation}
    \label{eq:6}
    \underset{(\U,\N,\V,\Ss) \sim \zeta_2}{\I}[\U^i_{-\V},\Pi''^i(\U)\ |\
    \N,\V,\Ss]  \leq  \frac{1}{\ell-1}\underset{(\X,\M,\W,\Z) \sim
      \zeta_\ell}{\I}[\X^i_{-\W},\Pi'^i(\X)\ |\ \M,\W,\Z],
  \end{equation}
and
  \begin{equation}
    \label{eq:7}
    \underset{(\U,\N,\V,\Ss) \sim \zeta_2}{\I}[\N_{-\V},\Pi''^i(\U)\ |\ \U^i,\V,\Ss] \leq 
    \frac{1}{\ell-1}\underset{(\X,\M,\W,\Z) \sim \zeta_\ell}{\I}[\M_{-\W},\Pi'^i(\X)\ |\ {\X}^i,\W,\Z].
  \end{equation}

\end{lemma}
\begin{proof}
  We consider the LHS of Equation \eqref{eq:7}. The view of processor
  $i$ of the transcript of protocol $\Pi''$, denoted as $\Pi''^i(\U)$,
  is given as follows.
  \begin{equation}
    \label{eq:8}
    \Pi''^i(\U)= \langle \Pp, \Qq, \Z_{-\Pp,\Qq}, \Rr, \Tt, \M_\Tt,
    \X^i_{\bar{\Tt}\backslash\{\Pp,\Qq\}}, \Pi'(\X(
  \Pp,\Qq, \U)) )\rangle.
  \end{equation}
So the LHS of Equation \eqref{eq:7} can be written as 
\begin{align}
  & \underset{(\U,\N,\V,\Ss) \sim \zeta_2}{\I}[\N_{-\V}:\Pi''^i(\U)\
  |\
  \U^i,\V,\Ss] \nonumber \\
  &= \underset{\underset{(\X,\M,\Z)\sim \zeta_\ell\downarrow_{\X,\M,\Z}}{(\U,\N,\V,\Ss) \sim \zeta_2}}{\I}[\N_{-\V}: \Pp, \Qq, \Z_{-\Pp,\Qq}, \Tt, \M_\Tt, \Rr,
  \X^i_{\bar{\Tt}\backslash\{\Pp,\Qq\}}, \Pi'^i(\X(
  \Pp,\Qq, \U))|\U^i, \V,\Ss]  \nonumber \\
  &= \underset{\underset{(\X,\M,\Z)\sim \zeta_\ell\downarrow_{\X,\M,\Z}}{(\U,\N,\V,\Ss) \sim \zeta_2}}{\I} [ \N_{-\V}: \Pi'^i(\X(
  \Pp,\Qq, \U))|\Pp, \Qq, \Z_{-\Pp,\Qq}, \Tt,
  \M_\Tt, \Rr, \X^i_{\bar{\Tt}\backslash\{\Pp,\Qq\}}, \U^i, \V,\Ss]
  \tag*{\text{[Lemma
      \ref{lemma:condn-excg} eqn. \eqref{eq:30}]}}\nonumber \\
  &=  \underset{\underset{(\X,\M,\Z)\sim \zeta_\ell\downarrow_{\X,\M,\Z}}{(\U,\N,\V,\Ss) \sim \zeta_2}}{\I}[\N_{-\V}: \Pi'^i(\X)\ |\ \Pp, \Qq, \Rr, \Tt, \M_\Tt, \Z, \V,
  \X^i_{\bar{\Tt}}] \tag*{\text{[Combining $(\U^i,
      \X^i_{\bar{\Tt}\backslash\{\Pp,\Qq\}})$
      and $(\Z_{-\Pp,\Qq},\Ss)$]}}\nonumber \\
  &\leq  \underset{\underset{(\X,\M,\Z)\sim \zeta_\ell\downarrow_{\X,\M,\Z}}{(\U,\N,\V,\Ss) \sim \zeta_2}}{\I}[\N_{-\V}: \Pi'^i(\X)\ |\ \Pp, \Qq, \Rr, \Tt, \M_\Tt, \Z,
  \V, \X^i]\tag*{\text{[Lemma \ref{lemma:condn-excg} eqn.
      \eqref{eq:31}, $\X^i_\Ss$ ind. of $\N_{-\V}$]}}  \nonumber
\end{align}

 \noindent [$V$ takes value in $1$ and $2$ uniformly at
     random. Hence we can write it as follows.]
\begin{align} 
\label{eq:9}
  =\frac{1}{2}&\underset{(\X,\M,\Z)\sim \zeta_\ell\downarrow_{\X,\M,\Z}}{\I}[\M_\Pp: \Pi'^i(\X)\ |\ \Pp, \Qq,\Rr, \Tt, \M_\Tt, \Z, \V=2,
  \X^i]\nonumber \\
  &+ \frac{1}{2}\underset{(\X,\M,\Z)\sim \zeta_\ell\downarrow_{\X,\M,\Z}
  }{\I}[\M_\Qq: \Pi'^i(\X)\ |\ \Pp,\Qq, \Rr, \Tt, \M_\Tt, \Z, \V=1,
  \X^i].
\end{align}
Consider the first mutual information term. 
\begin{align}
  \label{eq:10}
  & \I[\M_\Pp: \Pi'^i(\X)\ |\ \Pp, \Qq, \Rr, \Tt, \M_\Tt, \Z, \V=2,
  \X^i]\nonumber \\
  &=\frac{2}{\ell(\ell-1)}\sum_{p<q} \I[\M_p: \Pi'^i(\X)\ |\
  \Pp=p,\Qq=q, \Rr, \Tt, \M_\Tt, \Z, \V=2,
  \X^i] \nonumber  \\
  &=\frac{2}{\ell(\ell-1)^2}\sum_{p<q}
  \sum_{r=0}^{\ell-2}\sum_{t:|t|=r}\Pr[\Tt=t]\I[\M_p: \Pi'^i(\X)\ |\
  p, q, r, t, \M_t, \Z, \V=2,
  \X^i]\nonumber \\
  &=\frac{2}{\ell(\ell-1)^2}\sum_{p<q}
  \sum_{r=0}^{\ell-2}\sum_{t:|t|=r}\frac{(\ell-r-2)!r!}{(\ell-2)!}\I[\M_p:
  \Pi'^i(\X)\ |\ p,q,r,t,\M_t, \Z, \V=2,
  \X^i].\nonumber \\
  \intertext{We can safely drop the conditioning $\Pp=p, \Rr=r, \Tt=t$
    and $\V=2$ in the following way. It is easy to see $\Rr=r, \Tt=t$ is implied by
    $\M_t$. $\M_p$ implies $\Pp=p$. Moreover, given $(p,q)$, $\V=2$ is equivalent to $\W=p$. So we can write,} & \I[\M_\Pp: \Pi'^i(\X)\ |\ \Pp, \Qq, \Rr, \Tt, \M_\Tt, \Z,
  \V=2,
  \X^i]\nonumber \\
  &=\frac{2}{(\ell-1)!\ell(\ell-1)}\sum_q\sum_{p:p<q}
  \sum_{r=0}^{\ell-2}\sum_{t:|t|=r}((\ell-r-2)!r!)\I[\M_p: \Pi'^i(X)\ |\
  \W=q, \M_t, \Z,
  \X^i].
\end{align}
  Similarly, the second mutual information term of Equation
    \eqref{eq:9} term can be written in the following way. 
\begin{align}
  \label{eq:100}
 &\I[\M_\Qq: \Pi'^i(X)\ |\ \Pp, \Qq, \Rr, \Tt, \M_\Tt, \Z, \V=1,
  \X^i]\nonumber \\
  &=\frac{2}{(\ell-1)!\ell(\ell-1)}\sum_q\sum_{p:p<q}
  \sum_{r=0}^{\ell-2}\sum_{t:|t|=r}((\ell-r-2)!r!)\I[\M_q: \Pi'^i(\X)\ |\
  \W=p, \M_t, \Z,
  \X^i].
\end{align}

Combining Equation \eqref{eq:9}, \eqref{eq:10},
  \eqref{eq:100}, we get,
\begin{align}
  & \I_{(\U,\N,\V,\Ss) \sim
    \zeta_2}[\N_{-\V},\Pi'^i(\U)\ |\
  \U^i,\V,\Ss] \nonumber \\
  &\leq\frac{1}{(\ell-1)!\ell(\ell-1)}\sum_{q'}\sum_{p':p'\neq q'}
  \sum_{r=0}^{\ell-2}\sum_{t:|t|=r}((\ell-r-2)!r!)\I[\M_{p'}: \Pi'^i(X)|\W=q',
  \M_t, \Z,
  \X^i]\nonumber
  \intertext{The
number of permutations of $[\ell]\backslash q$ where the $r+1$th element is $p'$ and the first $r$ elements constitute the set $t$ is $(\ell-r-2)!r!$. Hence we can
write the previous summation as follows,}
  &=\frac{1}{(\ell-1)!\ell(\ell-1)}\sum_{q'}\sum_{\sigma \in
    \mathcal{S}_{[\ell]\backslash q'}}\sum_{i\in [\ell]\backslash
    q'}\I[\M_{\sigma(i)}: \Pi'^i(\X)\ |\ \M_{\{\sigma(1),...,\sigma(i-1)\}},
  \Z, \W=q',
  \X^i]\nonumber \\
  &= \frac{1}{(\ell-1)!\ell(\ell-1)} \sum_{q'}\sum_{\sigma \in
    \mathcal{S}_{[\ell]\backslash q'}}\I[\M_{-q'}: \Pi'^i(\X)\ |\ \Z, \W=q',
  \X^i]\tag*{[Using chain rule of information,
    Eq. \eqref{eq:49}]}\nonumber\\
  &= \frac{1}{\ell-1}\sum_{q'}\frac{1}{\ell}\underset{(\X,\M,\Z)\sim\zeta_\ell\downarrow_{\X,\M,\Z}}{\I}[\M_{-q'}: \Pi'^i(\X)\ |\ \Z, \W=q',
  \X^i]\nonumber \\
  &= \frac{1}{\ell-1}\underset{(\X,\M,\Z,\W)\sim\zeta_\ell}{\I}[\M_{-\W}: \Pi'^i(\X)\ |\ \Z, \W, \X^i].\label{eq:110}
\end{align}

Equation \eqref{eq:6} can be proved almost similarly. Let's look at the LHS. 
\begin{align}
  & \underset{(\U,\N,\V,\Ss) \sim \zeta_2}{\I}[\U^i_{-\V}:\Pi''^i(\U)\ |\ \N,\V,\Ss] \nonumber\\
  &= \underset{\underset{(\X,\M,\Z)\sim \zeta_\ell\downarrow_{\X,\M,\Z}}{(\U,\N,\V,\Ss) \sim \zeta_2}}{\I}[\U^i_{-\V}: \langle \Pp, \Qq, \Z_{-\Pp,\Qq}, \M_\Tt, \Rr, \Tt,
  \X^i_{\bar{\Tt}\backslash\{\Pp,\Qq\}}, \Pi'^i(\X(
  \Pp,\Qq, \U)) \rangle\ |\ \N, \V,\Ss]\nonumber \\
  &= \underset{\underset{(\X,\M,\Z)\sim \zeta_\ell\downarrow_{\X,\M,\Z}}{(\U,\N,\V,\Ss) \sim \zeta_2}}{\I} [ \U^i_{-\V}: \Pi'^i(\X(
  \Pp,\Qq, \U))\ |\ \Pp, \Qq, \Z, \M_\Tt, \Rr, \Tt,
  \X^i_{\bar{\Tt}\backslash\{\Pp,\Qq\}},  \N, \V,\Ss] \tag*{\text{[Lemma
    \ref{lemma:condn-excg} Equation \eqref{eq:30}]}}\nonumber \\
  &= \underset{\underset{(\X,\M,\Z)\sim \zeta_\ell\downarrow_{\X,\M,\Z}}{(\U,\N,\V,\Ss) \sim \zeta_2}}{\I}[\U^i_{-\V}: \Pi'^i(\X)\ |\ \Pp, \Qq, \Rr, \Tt, \M_{\Tt\cup\{p,q\}}, \Z, \V,
 \X^i_{\bar{\Tt}\backslash\{\Pp,\Qq\}}]\nonumber \\
  &\leq\underset{\underset{(\X,\M,\Z)\sim \zeta_\ell\downarrow_{\X,\M,\Z}}{(\U,\N,\V,\Ss) \sim \zeta_2}}{\I}[\U^i_{-\V}: \Pi'^i(\X)\ |\ \Pp, \Qq, \Rr, \Tt, \M, \Z, \V,
  \X^i_{\bar{\Tt}\backslash\{\Pp,\Qq\}}] \tag*{\text{[Lemma \ref{lemma:condn-excg} Equation
    \eqref{eq:31}]}}\nonumber \\ 
  &=\frac{1}{2}\underset{(\X,\M,\Z)\sim \zeta_\ell\downarrow_{\X,\M,\Z}}{\I}[\X^i_\Pp: \Pi'^i(\X)\ |\ \Pp, \Qq, \Rr, \Tt, \M, \Z, \V=2,
  \X^i_{\bar{\Tt}\backslash\{\Pp,\Qq\}}]\nonumber\\
  &+ \frac{1}{2}\underset{(\X,\M,\Z)\sim \zeta_\ell\downarrow_{\X,\M,\Z}}{\I}[\X^i_\Qq: \Pi'^i(\X)\ |\ \Pp,\Qq,
  \Rr, \Tt, \M, \Z, \V=1,
  \X^i_{\bar{\Tt}\backslash\{\Pp,\Qq\}}]. \label{eq:12}
\end{align}
Consider $\Tt'=\bar{\Tt}\backslash\{\Pp,\Qq\}$ which is distributed in the same
way as $\Tt$. Here the auxiliary variable for size of $\Tt'$ is $\Rr'=(\ell-2)-\Rr$
instead of $\Rr$. $\Rr'$ is distributed in the same  way as $\Rr$. We will
analyze, as before, the first term of \eqref{eq:12}.

 \begin{align}
   \label{eq:13}
   & \I[\X^i_\Pp: \Pi'^i(\X)\ |\ \Pp, \Qq, \Rr', \Tt', \M, \Z, \V=2,
   \X^i_{\Tt'}]\nonumber\\
   &=\frac{2}{\ell(\ell-1)}\sum_{p<q} \I[\X^i_p: \Pi'^i(\X)\ |\ \Pp=p, \Qq=q, \Rr', \Tt'  \M, \Z, \V=2,   \X^i_{\Tt'}] \nonumber \\
   &=\frac{2}{\ell(\ell -1)^2}\sum_{p<q}
   \sum_{r'=0}^{\ell -2}\sum_{t':|t'|=r'}\Pr[\Tt'=t']\I[\X^i_p: \Pi'^i(\X)\ |\
   p,q,r',t', \M, \Z, \V=2,
   \X^i_{t'}]\nonumber \\
   &=\frac{2}{\ell (\ell -1)^2}\sum_q\sum_{p:p<q}
   \sum_{r'=0}^{\ell -2}\sum_{t:|t|=r'}\frac{(\ell -r'-2)!r'!}{(\ell -2)!}\I[\X^i_p:
   \Pi'^i(X)\ |\ \W=q,\M, \Z,
   \X^i_{t'}].\nonumber \\
   \intertext{The last equality follows by dropping the implied
     conditioning and introducing random variable $\W$ to denote the
      all $1$ coordinate. As before we combine the second term of the Equation
     \eqref{eq:12} to get the following.}  
   &\I_{(\U,\N,\V,\Ss) \sim
     \zeta_2}[\U^i_{-\V}:\Pi''^i(\U)\ |\ \N,\V,\Ss] \nonumber \\
   &\leq \frac{1}{(\ell -1)!\ell (\ell -1)}\sum_{\underset{p\neq q}{(p.q):}}
   \sum_{r'=0}^{\ell -2}\sum_{t:|t|=r'}((\ell -r'-2)!r'!)\I[\X^i_p: \Pi'^i(\X)\ |\
   \W=q, \M, \Z, 
   \X^i_{t'}]\nonumber \\
   &= \frac{1}{(\ell -1)!\ell (\ell -1)}\sum_q \sum_{\sigma \in
     \mathcal{S}_{[\ell ]\backslash q}}\sum_{i\in [\ell ]\backslash
     q}\I[\X^i_{\sigma(i)}: \Pi'^i(\X)\ |\
   X^i_{\{\sigma(1),...,\sigma(i-1)\}}, \W=q, \Z,
   \M]\nonumber \\
   &= \frac{1}{(\ell -1)!\ell (\ell -1)} \sum_q \sum_{\sigma \in
     \mathcal{S}_{\ell -1}}\I[\X^i_{-\W}: \Pi'^i(\X)\ |\ \Z, \W=q,
   \M]\nonumber\\
   &= \frac{1}{\ell -1}\sum_q \frac{1}{\ell}\underset{(\X,\M,\Z)\sim \zeta_\ell\downarrow_{\X,\M,\Z}}{\I}[\X^i_{-\W}: \Pi'^i(\X)\ |\ \Z, \W=q,
   \M]\nonumber\\
   &= \frac{1}{\ell -1}\underset{(\X,\M,\Z,\W)\sim\zeta_\ell}{\I}[\X^i_{-\W}: \Pi'^i(\X)\ |\ \Z, \W, \M].
 \end{align}

Equation \eqref{eq:110} and Equation \eqref{eq:13} prove the lemma.
\end{proof}

Lemma \ref{lemma:mod-direct-sum} implies the Theorem \ref{thm:moddirectsum}.

\subsection{Lower bounding $\Disj$}
\label{sec:lower-bounding-disj}

In this section we prove the following lower bound for the partial information
complexity of $\Disj$.

\begin{theorem}
	\label{LB-DISJ2}
	$\mathsf{PIC}_{\zeta_2} = \Omega(k).$
\end{theorem}

This, combined with Theorem
\ref{thm:moddirectsum} and Theorem \ref{thm:direct-sum1} will imply a
$\Omega(m\ell k)$ lower bound on the switched information complexity of
$\tribes$ which is the lower bound on $R_\delta(\tribes)$ we aimed
for. 

\bigskip
\noindent
\textbf{Notation.}  By $\bar{e}$ we mean the all $1$ vector of size
$k$. By $\bar{e}_{i,j}$, we mean the boolean vector of
size $k$ where all entries are $1$ except the entries in index $i$ and
$j$. Similarly, $\bar{e}_i$ is the boolean vector the all entries are
$1$ except that of index $i$. $\Pi[i,x,m,z;\bar{e}_i]$ implies the
transcript of the protocol $\Pi$ on the following $\Disj$ instance:
the input of the first column comes from the distribution specified by
$\M=m$, $\Z=z$ and $\X_i=x$ and the input of the second column is
$\bar{e}_i$. Abusing notation slightly, $\Pi^i[x,m,z;\bar{e}_i]$
represents processor-$i$'s view of the transcript of the protocol
$\Pi$ when the input of the first column comes from the distribution
specified by $\X_i=x, \M=m$ and $\Z=z$ and the input of the second
column is $\bar{e}_i$.

\bigskip
\noindent
\textbf{Hellinger distance.}
For probability distributions $P$ and $Q$ supported on a sample space
$\Omega$, the Hellinger distance between $P$ and $Q$, denoted as
$h(P,Q)$, is defined as,
\begin{equation}
  \label{eq:40}
  h(P,Q) = \frac{1}{\sqrt{2}} ||\sqrt{P} - \sqrt{Q}||_2.
\end{equation}

\noindent Hellinger distance can also be written as follows,
\begin{equation}
  \label{eq:41}
  h(P,Q)^2=1-F(P,Q),
\end{equation}
where $F(P,Q) = \sum_{\omega \in \Omega}\sqrt{P(\omega)Q(\omega)}$ is
also known as Bhattacharya coefficient. From the definition it is
clear to see that the Hellinger distance is a metric satisfying
triangle inequality. Below  we will
state some facts (without proof) about Hellinger distance. Interested
readers can refer to \cite{DBLP:journals/jcss/Bar-YossefJKS04} for the
proofs.

We will denote the statistical distance between two distributions $P$
and $Q$ as $\Delta(P,Q)$ and the Hellinger distance between $P$ and
$Q$ as $h(P,Q)$.

\begin{fact}[Hellinger vs. statistical distance]
  \label{fac:StatDist-2}
\begin{equation}
  \label{eq:42}
  h(P,Q) \leq \Delta(P,Q) \leq \sqrt{2}h(P,Q).
\end{equation}
\end{fact}
This essentially means that the Hellinger distance is good
approximation of statistical distance. The following facts gives us the
necessary connection between mutual information and Hellinger
distance.

\begin{fact}[Hellinger vs information \cite{DBLP:journals/tit/Lin91}]
  \label{fac:mutInfo-HellDist}
 Let $X$ be a
random variable taking value in $\{x_1, x_2\}$ equally likely and let
$\Pi$ be a randomized protocol which takes $X$ as input. Then,
\begin{equation}
  \label{eq:43}
  \I[X:\Pi(X)] \geq h^2(\Pi(x_1),\Pi(x_2)). 
\end{equation}
\end{fact}

\begin{fact}
  \label{fac:StatDist-1}
Let $\Pi$ be a $\delta$-error protocol for function $f$. For inputs
$x$ and $y$ such that $f(x) \neq f(y)$, we have,
\begin{equation}
  \label{eq:75}
  h(\Pi(x), \Pi(y)) = \frac{1-\delta}{\sqrt{2}}.
\end{equation}
\end{fact}

The following lemmas will be helpful in our proof. These lemmas are
straight-forward generalization of their two-party analogues.
\begin{lemma}
  \label{lemma:cutpaste}
For any randomized protocol $\Pi$ computing $f: X^k  \rightarrow
\bool$ and for any $x, y \in X^k$ and for some $i$ and $j$, 
\begin{equation}
  \label{eq:16}
  h(\Pi(x_ix_jx_{-i,j}, y_iy_jy_{-i,j}))=h(\Pi(x_iy_jx_{-i,j}, y_ix_jy_{-i,j})). 
\end{equation}
\end{lemma}

\begin{proof}
  We will think of the randomized protocol $\Pi$ on input $(x,y)$ as a
  deterministic protocol $\Pi'$ working on $(x,y,\{R_i\}_i)$ where
  $R_i$ is the private random coins of player $i$. The first
  observation is that for $k$-party deterministic protocol, the inputs
  which gives rise to the same transcript $\pi$ form a combinatorial
  rectangle $\mathbf{R}_\pi= S^1_\pi \times \cdots \times S^k_\pi$. So
  we have the following.
  \begin{eqnarray}
    \label{eq:22}
    \Pr_{\{R_i\}_i}[\Pi'[x,\{R_i\}_i] = \pi] &=&
    \Pr_{\{R_i\}_i}[(x,\{R_i\}_i) \in \mathbf{R}_\pi] \cr
    &=& \Pr_{\{R_i\}_i}\left[\bigwedge_i (x_i,R_i) \in S^i_\pi\right] \cr
    &=& \bigwedge_i\Pr_{R_i}[(x_i,R_i) \in S^i_\pi].
  \end{eqnarray}

  Hence it immediately follows that
  \begin{align}
    \label{eq:23}
    & 1- h^2(\Pi(x_ix_jx_{-i,j}, y_iy_jy_{-i,j}))\nonumber \\
    &=
    \sum_\pi\sqrt{\Pr_{\{R_i\}_i}[\Pi(x_ix_jx_{-i,j}=\pi].\Pr_{\{R_i\}_i}[y_iy_jy_{-i,j}=\pi]}\cr
    &= \sum_\pi\bigg[\sqrt{\Pr_{\{R_i\}_i}[(x_{-i,j},R_{-i,j}\in
      \mathbf{R}_\pi^{-i,j}]\Pr_{R_i}[(x_i,R_i) \in
      S^i_\pi]\Pr_{R_j}[(x_j,R_j) \in S^j_\pi]}\times \nonumber \\
    & \ \ \ \ \ \ \ \ \ \ \ \ \ \ \ \ \ \ \ \ \ \ \ \ \ \ \ \ \ \
    \sqrt{\Pr_{\{R_i\}_i}[(x_{-i,j},R_{-i,j}\in
      \mathbf{R}_\pi^{-i,j}]\Pr_{R_i}[(y_i,R_i) \in
      S^i_\pi]\Pr_{R_j}[(y_j,R_j) \in S^j_\pi]}\bigg]\nonumber \\
    &=\sum_\pi\sqrt{\Pr_{\{R_i\}_i}[\Pi(x_iy_jx_{-i,j}=\pi].\Pr_{\{R_i\}_i}[y_ix_jy_{-i,j}=\pi]}\nonumber \\
    &= 1- h^2(\Pi(x_iy_jx_{-i,j}, y_ix_jy_{-i,j})).
  \end{align}
\end{proof}

This property of Hellinger distance is called the \textit{$k$-party
  cut-paste property}. Another property of Hellinger distance which is
required is the following.

\begin{lemma}
  \label{lemma:pythagorian}
  For any randomized protocol $\Pi$ and for any input $x, y \in X^k$
  and for some $i$ and $j$,
  \begin{equation}
    \label{eq:17}
    2h^2(\Pi(x_ix_jx_{-i,j}, y_iy_jy_{-i,j})) \geq h^2(\Pi(x_ix_jx_{-i,j},
    x_iy_jy_{-i,j})) + h^2(\Pi(y_ix_jx_{-i,j}, y_iy_jy_{-i,j})).
  \end{equation}
\end{lemma}

\begin{proof}
  As before,
  \begin{align}
    \label{eq:24}
    &  (1-h^2(\Pi(x_ix_jx_{-i,j},
    x_iy_jy_{-i,j})))+(1-h^2(\Pi(y_ix_jx_{-i,j}, y_iy_jy_{-i,j})))
    \nonumber \\
    &= \sum_\pi
    \bigg[\sqrt{\Pr_{\{R_i\}_i}[\Pi(x_ix_jx_{-i,j}=\pi].\Pr_{\{R_i\}_i}[x_iy_jy_{-i,j}=\pi]}
      +\nonumber\\
    &\ \ \  \ \ \ \ \ \ \ \ \ \ \ \ \ \ \ \ \ \ \ \ \ \ \ \ \ \ \ \ \
    \ \ \ \ \ \ \ \ \ \ \ \ \   \sqrt{\Pr_{\{R_i\}_i}[\Pi(y_ix_jx_{-i,j}=\pi].\Pr_{\{R_i\}_i}[y_iy_jy_{-i,j}=\pi]}\bigg]
    \nonumber \\
    &=\sum_\pi\bigg[\sqrt{\Pr_{\{R_i\}_i}[(y_{-i,j},R_{i,j}\in
        \mathbf{R}_\pi^{-i,j}]\Pr_{\{R_i\}_i}[(x_{-i,j},R_{i,j}\in
        \mathbf{R}_\pi^{-i,j}]}\times \nonumber\\
    &\ \ \  \ \ \ \ \ \ \ \ \ \sqrt{\Pr_{R_j}[(x_j,R_j) \in
        S^j_\pi]\Pr_{R_j}[(y_j,R_j) \in S^j_\pi]}\left(\Pr_{R_i}[(x_i,R_i) \in
          S^i_\pi] + \Pr_{R_i}[(y_i,R_i) \in
          S^i_\pi]\right)\bigg]\nonumber \\
    &\geq 2\sum_\pi\bigg[\sqrt{\Pr_{\{R_i\}_i}[(y_{-i,j},R_{i,j}\in
        \mathbf{R}_\pi^{-i,j}]\Pr_{\{R_i\}_i}[(x_{-i,j},R_{i,j}\in
        \mathbf{R}_\pi^{-i,j}]}\times\nonumber \\
    &\ \ \  \ \ \ \ \ \ \ \ \ \ \ \ \ \sqrt{\Pr_{R_j}[(x_j,R_j) \in
        S^j_\pi]\Pr_{R_j}[(y_j,R_j) \in S^j_\pi]}\sqrt{\Pr_{R_i}[(x_i,R_i) \in
        S^i_\pi]\Pr_{R_i}[(y_i,R_i) \in
        S^i_\pi]}\bigg] \nonumber \\
    &= 2(1-h^2(\Pi(x_ix_jx_{-i,j}, y_iy_jy_{-i,j}))).
  \end{align}
\end{proof}

Now we show some structural properties of coordinator model described
in terms of Hellinger distance. These are generalizations of analogous
properties shown in \cite{DBLP:conf/focs/BravermanEOPV13}. We believe,
the proofs are simpler than the ones in
\cite{DBLP:conf/focs/BravermanEOPV13}.  We will make use of these
structural properties later in the proof.

\subsubsection{Structural properties}
\label{sec:struct-prop}

First we state a version of \textit{diagonal} lemma for $\M$ and
$\X^i$ which will be useful in our proof. Note that
$\Pi^i(\bar{e}_{i,j}, \bar{e})$ is equivalent to $\Pi^i[1,1,j;\bar{e}]$.

\begin{lemma}
  \label{lemma:diagonal}
For $i \neq j$
\begin{equation}
  \label{eq:18}
  h^2(\Pi^i[0,0,j;\bar{e}], \Pi^i[1,1,z;\bar{e}])  
  \geq \frac{1}{2} h^2(\Pi^i(\bar{e}_{i,j};\bar{e}), \Pi^i(\bar{e}_j;\bar{e})).
\end{equation}
\end{lemma}

The next lemma is, as mentioned in
\cite{DBLP:conf/focs/BravermanEOPV13}, a version of
\textit{global-to-local} property of Hellinger distance in the
following setting.

\begin{lemma}
  \label{lemma:glob2loc}
For $i \neq j$
\begin{equation}
  \label{eq:19}
  h^2(\Pi[i,0,0,z;\bar{e}], \Pi[i,1,0,z;\bar{e}]) = h(\Pi^i[0,0,z;\bar{e}], \Pi^i[1,0,z;\bar{e}]),
\end{equation}
and
\begin{equation}
  \label{eq:20}
  h(\Pi(\bar{e}_{i,j};\bar{e}), \Pi(\bar{e}_i;\bar{e}) = h(\Pi^i(\bar{e}_{i,j};\bar{e}), \Pi^i(\bar{e}_i;\bar{e})).
\end{equation}
\end{lemma}
The proofs of Lemma \ref{lemma:diagonal} and Lemma
\ref{lemma:glob2loc} are straightforward from the analogous lemmas in
\cite{DBLP:conf/focs/BravermanEOPV13}. We will give a simplified
version of the proofs here.

\begin{proof}[Proof of Lemma \ref{lemma:diagonal}] We have to measure
  the Hellinger distance of the distribution of $\Pi^i$ where the
  inputs in the first coordinate are coming from the distribution
  switched by $\M$ and $\Z$. The first observation that we do is the
  following. In the protocol $\Pi$ the input of the $i$th player is
  fixed. The protocol can be thought of as a two party protocol, where
  the first player is the $i$-th player and the second player consists
  of all the $k-1$ players and the coordinator. The second player,
  given $\M$ and $\Z$, can randomly sample the inputs of all other
  players participating in $\Pi$. The input of the first player, as in
  the case of $i$th player in the protocol $\Pi$ is fixed. Let us call
  this new protocol as $\hat{\Pi}$. It is clear that the Hellinger
  distance in the Equation \eqref{eq:18} remains the same if we
  consider $\hat{\Pi}$ instead of $\Pi$.

  Now, given a two party protocol $\hat{\Pi}$, we can invoke
  Pythagorean lemma (two party version of Lemma
  \ref{lemma:pythagorian}) to imply the following.

  \begin{align}
    \label{eq:21}
2h^2(\Pi[0,0,j;\bar{e}], \Pi[1,1,j;\bar{e}])
    &=2h^2(\hat{\Pi}[0,(0,j);\bar{e}],
    \hat{\Pi}[1,(1,j);\bar{e}])\nonumber \\
    &\geq h^2(\hat{\Pi}[0,(0,j);\bar{e}],
    \hat{\Pi}[1,(0,j);\bar{e}])\nonumber \\
    &\ \ \ \ \ \ \ \ \ \ \ \ \ +h^2(\hat{\Pi}[0,(1,j);\bar{e}],
    \hat{\Pi}[1,(1,j);\bar{e}])\nonumber \\
    &\geq h^2(\hat{\Pi}[0,(1,j);\bar{e}],
    \hat{\Pi}[1,(1,j);\bar{e}])\nonumber \\
    &= h^2(\Pi[\bar{e}_{i,j};\bar{e}],
    \hat{\Pi}[\bar{e}_j;\bar{e}]).\nonumber
  \end{align}
\end{proof}

\begin{proof}[Proof of Lemma \ref{lemma:glob2loc}]
  We will make the following observations. Firstly, given any
  transcript $\tau$ for $\Pi$ and a player $i$ we can divide $\tau$ in
  three parts: $\tau_{i\leftarrow}$ which is the part of the
  transcript where the coordinator sends message to player $i$,
  $\tau_{i\rightarrow}$ which is the part of the transcript where the
  player $i$ sends message to the coordinator and $\tau_{-i}$ where
  the other players and the coordinator sends message to each other.

  Secondly, the following is easy to see. When the input of the player
  $i$ is fixed (say to $(01)$) as in the case of the distribution we
  are interested in Eq. (\ref{eq:19}),
  \begin{equation}
    \label{eq:28}
  \Pr_{R,X}[\Pi=\tau]=\Pr_{R_i}[\Pi_{i\rightarrow}=\tau_{i\rightarrow}\
  |\
  X_i=(01),\Pi_{i\leftarrow}=\tau_{i\leftarrow}].\Pr_{\underset{R_{-i}}{X_{-i}}}[\Pi_{i\leftarrow}=\tau_{i\leftarrow},\Pi_{-i}=\tau_{-i}\
  |\ \Pi_{i\rightarrow}=\tau_{i\rightarrow}].
  \end{equation}

  Similarly,
  \begin{equation}
    \label{eq:29}
    \Pr_{R,X}[\Pi^i=\tau^i]=\Pr_{R_i}[\Pi^i_{i\rightarrow}=\tau^i_{i\rightarrow}\
  |\
  X_i=(01),\Pi^i_{i\leftarrow}=\tau^i_{i\leftarrow}].\Pr_{R_{-i},X_{-i}}[\Pi^i_{i\leftarrow}=\tau^i_{i\leftarrow}\
|\ \Pi^i_{i\rightarrow}=\tau^i_{i\rightarrow}].
  \end{equation}
  So,
  \begin{align}
    \label{eq:27}
     1&- h^2(\Pi^i[0,0,z;\bar{e}], \Pi^i[1,0,z;\bar{e}])\nonumber  \\
     &=\sum_{\tau^i}\bigg[\sqrt{\Pr_{R_i}[\Pi^i_{i\rightarrow}=\tau^i_{i\rightarrow}\
         |\ 
         X_i=(01),
\Pi^i_{i\leftarrow}=\tau^i_{i\leftarrow}]
\Pr_{R_i}[\Pi^i_{i\rightarrow}=\tau^i_{i\rightarrow}\
        |\ X_i=(11),\Pi^i_{i\leftarrow}=\tau^i_{i\leftarrow}]}\nonumber\\
     &\ \ \ \  \ \ \ \ \   \Pr_{R_{-i},X_{-i}}[\Pi^i_{i\leftarrow}=\tau^i_{i\leftarrow}\
 |\ \Pi^i_{i\rightarrow}=\tau^i_{i\rightarrow}]\bigg]\nonumber\\
 &=\sum_{\tau^i}\bigg[\sqrt{\Pr_{R_i}[\Pi^i_{i\rightarrow}=\tau^i_{i\rightarrow}\
         |\ 
         X_i=(01),
\Pi^i_{i\leftarrow}=\tau^i_{i\leftarrow}]
\Pr_{R_i}[\Pi^i_{i\rightarrow}=\tau^i_{i\rightarrow}\
        |\ X_i=(11),\Pi^i_{i\leftarrow}=\tau^i_{i\leftarrow}]}\nonumber\\
     &\ \ \ \  \ \ \ \ \  \sum_{\tau:\tau|_{i}=\tau^i} \Pr_{R_{-i},X_{-i}}[\Pi_{i\leftarrow}=\tau_{i\leftarrow},\Pi_{-i}=\tau_{-i}\
 |\ \Pi_{i\rightarrow}=\tau_{i\rightarrow}]\bigg]\nonumber\\
 &=\sum_{\tau}\bigg[\sqrt{\Pr_{R_i}[\Pi_{i\rightarrow}=\tau_{i\rightarrow}\
         |\ 
         X_i=(01),
\Pi_{i\leftarrow}=\tau_{i\leftarrow}]
\Pr_{R_i}[\Pi_{i\rightarrow}=\tau_{i\rightarrow}\
        |\ X_i=(11),\Pi_{i\leftarrow}=\tau_{i\leftarrow}]}\nonumber\\
     &\ \ \ \  \ \ \ \ \   \Pr_{R_{-i},X_{-i}}[\Pi_{i\leftarrow}=\tau_{i\leftarrow},\Pi_{-i}=\tau_{-i}\
 |\ \Pi_{i\rightarrow}=\tau_{i\rightarrow}]\bigg]\nonumber\\
&=1-h^2(\Pi[i,0,0,z;\bar{e}], \Pi[i,1,0,z;\bar{e}]).\nonumber  
 \end{align}
\eqref{eq:20} can be proved in the similar way.
\end{proof}

\subsubsection{Information complexity of $\Disj$}
\label{sec:inform-compl-disj}

Now we are ready to prove the partial information cost of $\Disj$ is
$\Omega(k)$. We consider processor $i$ and fix a value $j \neq i$.



 \begin{claim}[\cite{DBLP:conf/focs/BravermanEOPV13}]
\label{claim}
   \begin{align}
    (1)&\I[\M_{-\W}:\Pi^i\ |\ \X^i,\Z=j,\W=2] \geq \frac{2}{3}
     h^2(\Pi^i[1,0,j;\bar{e}],\Pi^i[1,1,j;\bar{e}]),\\
    (2)&\I[\X^i_{-\W}:\Pi^i\ |\ \M, \Z=j,\W=2] \geq \frac{2}{3}
     h^2(\Pi^i[0,0,j;\bar{e}],\Pi^i[1,0,j;\bar{e}]),\\
     (3)&\I[\M_{-\W}:\Pi^i\ |\ \X^i,\Z=j,\W=1] \geq \frac{2}{3}
     h^2(\Pi^i[\bar{e};1,0,j],\Pi^i[\bar{e};1,1,j]),\\
    (4)&\I[\X^i_{-\W}:\Pi^i\ |\ \M, \Z=j,\W=1] \geq \frac{2}{3}
     h^2(\Pi^i[\bar{e};0,0,j],\Pi^i[\bar{e};1,0,j]).
   \end{align}
 \end{claim}

\begin{proof}[Proof sketch.]
  Note that $\M$ takes value $0$ with probability $2/3$ and $1$ with
  probability $1/3$. This makes $\M_{-\W}$ a completely unbiased
  random variable given $\X^i=(11)$ and $\Z=j$. This lets us make the
  following assertion from the property of Hellinger distance
  (c.f. Fact \ref{fac:mutInfo-HellDist}).
\begin{equation}
  \label{eq:35}
  \I[\M_{-\W}:\Pi^i\ |\ \X^i=(11),\Z=j,\W=2] \geq 
     h^2(\Pi^i[1,0,j;\bar{e}],\Pi^i[1,1,j;\bar{e}]).
\end{equation}

Also, given $\M=0$ and $\Z=j$, $\X^i$, for any $i \neq j$, is a random
variable taking value between $(01)$ and $(11)$ uniformly. This lets
us conclude the following using Fact \ref{fac:mutInfo-HellDist}.

\begin{equation}
  \label{eq:36}
  \I[\X^i_{-\W}:\Pi^i\ |\ \M=0,\Z=j,\W=2] \geq 
     h^2(\Pi^i[1,0,j;\bar{e}],\Pi^i[0,0,j;\bar{e}]).
\end{equation}

Now, it can be checked that $\Pr[\X^i_{-\W}=(11)\ |\ \Z=j,\W=2] = 2/3$ and
by our construction $\Pr[\M=0\ |\ \Z=j, \W=2] = 2/3$. Combining these
facts with Equation \eqref{eq:35} and \eqref{eq:36} we can prove the
claim. Other two cases can be proved similarly.

\end{proof}

\noindent
Using Cauchy-Schwarz and triangle inequality, we can
write the following.
\begin{align}
  \I&[\M_{-W}:\Pi^i\ |\ \X^i, \Z=j, \W] + \I[\X^i_{-\W};\Pi^i\ |\
  \M,\Z=j,\W]\nonumber \\
  &\geq \frac{1}{3}\big[h^2(\Pi^i[1,1,j;\bar{e}],\Pi^i[0,0,j;\bar{e}])
  + h^2(\Pi^i[\bar{e};1,1,j],\Pi^i[\bar{e};0,0,j])\big]\nonumber \\
  \intertext{Using Lemma \ref{lemma:diagonal},}
 & \geq \frac{1}{6}\big[h^2(\Pi^i(\bar{e}_{i,j}.\bar{e}),\Pi^i(\bar{e}_j.\bar{e}))
  + h^2(\Pi^i(\bar{e}\bar{e}_{i,j}),\Pi^i(\bar{e}\bar{e}_j)\big].\nonumber
\end{align}
Using Lemma \ref{lemma:glob2loc} we get,
\begin{align}
  \label{eq:26}
  \sum_i&\I[\M_{-\W}:\Pi^i\ |\ \X^i, \Z, \W] + \I[\X^i_{-\W};\Pi^i\ |\
  \M,\Z,\W]
  \nonumber \\
  &\geq \frac{1}{6k}\sum_i\sum_{j:i \neq
    j}\Big[h^2(\Pi(\bar{e}_{i,j}.\bar{e}),\Pi(\bar{e}_j.\bar{e}))
  + h^2(\Pi(\bar{e}\bar{e}_{i,j}),\Pi(\bar{e}\bar{e}_j))\Big]\nonumber \\
  \intertext{By recounting the double summation,} \nonumber &=
  \frac{1}{12k} \sum_{i \neq
    j}\Big[[h^2(\Pi(\bar{e}_{i,j}.\bar{e}),\Pi(\bar{e}_j.\bar{e}))+h^2(\Pi(\bar{e}_{i,j}.\bar{e}),\Pi(\bar{e}_i.\bar{e}))]
  \nonumber \\
  &\ \ \ \ \ \ \ \ \ \ \ \ \ \ \ \ \ \ +
  [h^2(\Pi(\bar{e}\bar{e}_{i,j}),\Pi(\bar{e}\bar{e}_j))+h^2(\Pi(\bar{e}\bar{e}_{i,j}),\Pi(\bar{e}\bar{e}_i))]\Big]\nonumber
  \\
  \intertext{Using Cauchy-Schwarz \& triangle inequality,} &\geq
  \frac{1}{24k} \sum_{i \neq
    j}\Big[[h^2(\Pi(\bar{e}_{i}.\bar{e}),\Pi(\bar{e}_j.\bar{e}))] +
  [h^2(\Pi(\bar{e}\bar{e}_i),\Pi(\bar{e}\bar{e}_j))]\Big] \nonumber\\
  &= \frac{1}{24k} \sum_{i \neq
    j}\Big[[h^2(\Pi(\bar{e}.\bar{e}),\Pi(\bar{e}_{i,j}.\bar{e}))] +
  [h^2(\Pi(\bar{e}\bar{e}),\Pi(\bar{e}\bar{e}_{i,j}))]\Big]\tag*{\text{[Lemma~\ref{lemma:cutpaste}]}}\nonumber \\
  &\geq \frac{1}{48k} \sum_{i \neq
    j}[h^2(\Pi(\bar{e}.\bar{e}_{i,j}),\Pi(\bar{e}_{i,j}.\bar{e}))]\tag*{\text{[Cauchy-Schwarz \& triangle inequality]}}\nonumber \\
  &\geq \frac{1}{96k} \sum_{i \neq
    j}\Big[[h^2(\Pi(\bar{e}.\bar{e}_{i,j}),\Pi(\bar{e}_j.\bar{e}_i))+h^2(\Pi(\bar{e}_{i,j}\bar{e}),\Pi(\bar{e}_i.\bar{e}_j))
  ]\Big]\tag*{\text{[Lemma \ref{lemma:pythagorian}]}}\nonumber \\
  &= \frac{k-1}{384}(1-\delta)^2 \tag*{\text{[Fact \ref{fac:StatDist-1}]}}\nonumber \\
  &= \Omega(k).
\end{align}

We can write the penultimate equality because $\bar{e}.\bar{e}_{i,j}$
and $\bar{e}_{i,j}\bar{e}$ gives output $1$ in $\Disj$ but
$\bar{e}_j.\bar{e}_{i}$ and $\bar{e}_i.\bar{e}_{j}$ gives $0$.


\section{Information cost \& communication}
\label{sec:corr-inform-meas}
In this section we will show the information complexity is
right measure to lower bound by showing that the randomized
communication complexity of any function $f$ is lower bounded by the
switched information complexity of $f$.

\begin{theorem}
  \label{thm:info-meas}
For any distribution $\mu$ over the inputs,
\begin{equation}
  \label{eq:39}
  R_\epsilon(\tribes) = \Omega(\mathsf{IC}_\mu(\tribes)).
\end{equation}
\end{theorem}

\begin{proof}

Let us assume the random variables $(\X,\Z)$ is distributed according to
$\mu$ and the marginal distribution on $\X$ is $\nu$. Note that
$\I_\mu[\X;\Y\ |\ \Z] \leq \h_\mu(\X\ |\ \Z) \leq \h_\nu(\X)$. Now consider
any $\epsilon$-error protocol $\Pi$ for $\tribes$. We can write the
following.

 \begin{align}
   \label{eq:11}
   \I[\X^i:\Pi^i(\X)\ |\ \M,\Z] &\leq \h(\Pi^i\ |\ \M,\Z) \leq \h(\Pi^i),\\
   \intertext{and}
   \I[\M:\Pi^i(\X)\ |\ \X^i,\Z] &\leq \h(\Pi^i\ |\ \X^i,\Z) \leq \h(\Pi^i).
 \end{align}

Now trivially $\h(\Pi^i)$ upper bounded by the biggest size of
$\Pi^i$ (Note that, $\Pi^i$ is function of the random variable $X$)
and thus the Equation \eqref{eq:11} can be upper bounded by the
biggest size of $\Pi^i$. But for each player the biggest size of his
view of transcript can occur for different $X$. Hence we cannot upper
bound the switched information complexity by the $|\Pi|$.

Instead, we use the following fact from information theory.

\begin{lemma}[Theorem 5.3.1 in \cite{DBLP:books/daglib/0016881}]
\label{lemma:inst-code}
  The expected length $L$ of any instantaneous $q$-ary code for a
  random variable $\X$ satisfies the following inequality.
  \begin{equation}
    \label{eq:37}
    L \geq \frac{1}{\log q} \h(\X).
  \end{equation}
\end{lemma}

We can make the transcript instantaneous by introducing a special
delimiter . That still keeps the alphabet size constant. Hence we can
write the following.

 \begin{align}
   \label{eq:14}
   \I[\X^i:\Pi^i(\X)\ |\ \M,\Z] &\leq \h(\Pi^i\ |\ \M,\Z) \leq \h(\Pi^i)
   \leq \log 3.\mathbb{E}(|\Pi^i|)),\\
   \intertext{and} \I[\M:\Pi^i(\X)\ |\ \X^i,\Z] &\leq \h(\Pi^i\ |\ \X^i,\Z)
   \leq \h(\Pi^i) \leq \log 3.\mathbb{E}(|\Pi^i|).
 \end{align}

Now we are in good shape. We will complete the proof of the lemma by
the following set of equations.

\begin{align}
  \label{eq:38}
  \sum_{i \in [k]} (\I[\X^i:\Pi^i(\X)\ |\ \M,\Z] + \I[\M;\Pi^i(\X)\ |\
  \X^i,\Z] ) &\leq 2\log 3\sum_{i \in [k]}\mathbb{E}(|\Pi^i|) \nonumber \\ 
  &= 2\log 3 \mathbb{E}_x(\sum_{i \in [k]}|\Pi^i|)\tag*{\text{[Linearity of expectation]}}\nonumber \\
  &= 2\log 3\mathbb{E}_x(|\Pi|)\nonumber \\
  &= O(\max_x\{|\Pi(x)|\}). 
\end{align}

The worst case transcript size upper bounded by the randomized
communication complexity of $\tribes$. This proves the
theorem.

\end{proof}

Combining Theorem \ref{thm:info-meas}, \ref{thm:direct-sum1}, \ref{thm:moddirectsum} and \ref{LB-DISJ2}, we can prove Theorem \ref{thm:tribes}.

\bibliography{tribes}

\end{document}